\documentclass[11pt]{article}
\pdfoutput=1
\usepackage{amsmath, amssymb, amsthm, verbatim}
\usepackage{subfig}
\usepackage{fancyvrb}
\usepackage{epsfig}
\usepackage{rotating}

\def \figurespath {figs/} 
\def \figextension {pdf}

\newif\ifcompiletikz
\usepackage{tikz}
\usetikzlibrary{patterns}
\usetikzlibrary{external}
\ifcompiletikz
\usepackage{pgfplots}
  \newcommand{\inputtikz}[1]{
      \tikzsetnextfilename{#1}   
      \input{\figurespath#1.\tikzextension}
  }
\pgfplotsset{compat=newest}
\else 
  \newcommand{\inputtikz}[1]{\includegraphics{\figurespath#1.\figextension}}
\fi

\newtheorem{theorem}{Theorem}

\newtheorem{lemma*}[theorem]{Lemma}
\newtheorem{definition}[theorem]{Definition}

\newtheorem{proposition}[theorem]{Proposition}

\newtheorem{claim}{Claim}

\def\C{\mathcal{C}}

\def\F{\mathbb{F}}

\def\phi{\varphi}

\begin{document}

\title{Irregular Product Codes}
\author{Masoud Alipour, Omid Etesami, Ghid Maatouk, Amin Shokrollahi}
\maketitle

\begin{abstract}
We consider irregular product codes.
In this class of codes, each codeword is represented by a matrix.
The entries in each row (column) of the matrix should come from a component row (column) code.
As opposed to (standard) product codes, 
we do not require that all component row codes nor all component column codes be the same.
As we will see, relaxing this requirement can provide some additional attractive features
including 
1) allowing some regions of the codeword be more error-resilient 
2) allowing a more refined spectrum of rates for finite-lengths and improved performance in some of these rates
3) more interaction between row and column codes during decoding.

We study these codes over erasure channels.
We find that for any $0 < \epsilon < 1$, 
for many rate distributions on component row codes,
there is a matching rate distribution on component column codes 
such that an irregular product code based on MDS codes with those rate distributions on the component codes has asymptotic rate $1 - \epsilon$ and can decode
on erasure channels (of alphabet size equal the alphabet size of the component MDS codes)
with erasure probability $< \epsilon$.
\end{abstract}

\section{Introduction}

Product codes were introduced in 1954 by Elias \cite{elias_error-free_1954}. 
A product code can be viewed as a special case of Tanner construction
\cite{tanner_recursive_1981} in which smaller constituent codes make a larger
code with low complexity decoding.
An $m \times n$ product code is defined by a \textit{row code} $C$ of
length $n$ and rate $r_C$, and a \textit{column code} $C'$ of length $m$ and
rate $r_{C'}$. Codewords are represented by $m \times n$ matrices which
satisfy the constraint that every row belongs to $C$ and every column to
$C'$. Product codes are decoded in an iterative fashion, where rows and
columns are recovered in successive rounds using the decoders for $C$ and
$C'$. The rate of the product code is the product of the rates $r_C$ and
$r_{C'}$.

In this work, we present irregular product codes, a generalization of product
codes in which we do not require that the rows (columns) belong to a single
code. We will show that while these codes still retain the advantages of
product codes, they present some additional attractive features.

One of the main advantages of product codes is the fact that decoding takes
place over the smaller component codes, which can result in a speedup
of decoding. Furthermore, 
by
combining Reed-Solomon component codes, one can obtain product codes which
have length equal to the square of the size of the component codes for the same field size,
while taking advantage of the MDS properties of the small component codes.

Another (more application-specific) feature of product  codes is that they
perform well on bursty channels. Indeed, for a product code which is transmitted
row by row, a burst error will corrupt several consecutive rows but spread
evenly over columns, thus allowing the column codes to recover the corrupted
entries.

Irregular product codes are based on the simple idea that we need not restrict
ourselves to a single row and column code, but instead allow row and column
codes of multiple rates. The intuition behind this is that allowing for a few
low-rate, highly error-resilient codes might boost the decoding process, while
other high-rate codes ensure that the overall irregular product code has good
rate. With a careful design of the rate distributions, one can hope to achieve
better performance than for regular product codes. 
Irregularity has been a powerful concept in many contexts; e.g., irregular
degree distributions for LDPC codes, LT codes, etc.
This idea fully exploits the inherently interactive nature of the decoding of
product codes. Indeed, round-based decoding of product codes lets some rows and
columns ``help'' others to recover and go on with the decoding process. 
Allowing
for various decoding capabilities for different rows and columns only taps
further into this property of the decoder. 
\footnote{Indeed, in product codes that achieve rates close to Shannon limit (say on erasure channels),
either the row/column code (say row code) should have rate close to 1.
In this case, the decoding happens first in the column codes whose rate is far from 1,
and then the row codes play a ``complementary" role.
As we will see, there exist irregular product codes with rate vs.\ decoding capacity matching
these product codes in which the row codes and column codes have
the same distribution of rates, and the decoding process involves a longer and
gradual interaction between row and column component codes.}

Irregular product codes retain the advantages of product codes, while
presenting additional features that make them more attractive. Decoding still
takes place over smaller codes and the field size is still allowed to grow
slower in the case of MDS component codes. Further, not only do irregular
product codes still perform well on bursty channels, they can also be more
powerful than regular product codes when some parts of the codeword are known to
be more vulnerable to bursts than others, since the row and column codes
error-correction capabilities are tunable.

Moreover, for short-length linear codes, there do not exist product codes of
every desirable dimension, since fixing the dimension of the product code leaves
few choices for the dimensions of the component codes. Irregular product codes,
on the other hand, allow for many more dimensions due to the numerous choices
for the rate distribution of the component codes.

In this work, we first derive bounds on the rate and minimum distance
of irregular product codes, and give constructions that achieve these bounds.
We then give explicit families of irregular product codes that can get
rates arbitrarily close to $1 - \epsilon$ on channels with erasure
$\epsilon$ based on MDS component codes. 
Note however that this does not mean that these codes are capacity-approaching in the sense of Shannon capacity 
because the field size for MDS codes can grow as a function of the length.\footnote{On the other hand,
one can show that our analysis can be extended to the situation where instead of MDS codes as component codes,
we use capacity-approaching codes of the same rate but
over a fixed erasure channel, say BEC. 
In this case, the resulting product code will be truly capacity-approaching.}

We give simulation results
for finite-length codes that show that irregular product codes have
better thresholds than product codes of the same dimension or close dimension
for some specific lengths.

\subsection{Related works}
Since the introduction of the product codes \cite{elias_error-free_1954} many
extensions have been proposed and these codes have found many applications from
magnetic recording \cite{chaichanavong_tensor-product_2006} to deep space
communication \cite{baldi_class_2009} mainly because of their simple construction and low
complexity decoding.

The use of different component codes for rows and different component codes for columns 
is not new. 
In fact, \cite{stankovic_joint_2002} and \cite{cao_novel_2003} consider product codes
for image transmission
where the rows are LDPC codes and the columns are RS codes with different rates. They determine the optimum rate of the RS codes by a dynamic
programming.

However, to the best of our knowledge,
irregular product codes with the generality considered in this paper 
together with some of their asymptotic behavior have not been previously similarly explored.


Multidimensional product codes are investigated in \cite{rankin_single_2001} 
and \cite{rankin_asymptotic_2003}. However, the component codes are restricted
to be single parity and extended Hamming codes. In these papers, the authors devise a low complexity
soft decoding algorithm for AWGN channels.

The weight distribution of some instances of product codes is known. For
example, \cite{chiaraluce_extended_2004} analyzes the error floor region of an
extended Hamming product code by means of the weight enumerator of the code and the
 union bound. Some characterization of the stopping sets over the erasure
channel is obtained in \cite{rosnes_stopping_2008} based on the minimum
distance of the component codes. \cite{blankenship_block_2005} tries to optimize
the design of a product code where the component codes are limited to
single parity codes and certain extended Hamming and BCH codes. 

Product codes 
can be decoded iteratively using a message passing algorithm in noisy channels.
Because of this, they are also referred to as turbo block codes \cite{pyndiah_near-optimum_1998}
in the literature of coding theory.

The Tanner graph of the product code is regular. \cite{lentmaier_product_2010}
considers product codes as structured generalized LDPC codes. 

For a thorough survey on product codes refer to \cite{kschischang_product_2003}.

\subsection{Organization of the Paper}
The remainder of the paper is organized as follows. 
In Section~\ref{sec:Def}, we define
irregular product codes.
In Section~\ref{sec:Rate}, we derive an upper bound on their dimension, and prove
that under certain conditions, this upper bound can be achieved. 
In Section~\ref{sec:Dist}, we also derive
a lower bound on the minimum distance of irregular product codes and 
show that sometimes this lower bound is achieved.
In
Section~\ref{asymptoticSection}, we turn to the asymptotic analysis of irregular product codes on erasure channels
under the iterative decoding which switches back and forth between rows and columns.
In Section~\ref{sec:Cap}, 
we give explicit
families of irregular product codes based on MDS component
codes that achieve rates close to what capacity-achieving codes achieve. 
Finally, in Section~\ref{sec:Finite}, we give some irregular product code
constructions for specific code lengths and show by simulation that these
constructions outperform regular product code of the same (or approximately the
same) dimension.

\section{Definition}\label{sec:Def}
We denote the set $\{1, \ldots, m\}$ by $[m]$.

\begin{definition}\label{def:gpc}
Let $\F$ be a field and let $m, n$ be positive integers.
For each $i  \in [m]$ let $C_i$ be a code of length $n$ over $\F$ 
and  for each $j \in [n]$ let $C'_j$ be a code of length $m$ over $\F$.

The $m \times n$ \textit{irregular product code} $\C = \C(\{C_i\}_i,\{C'_j\}_j)$
is the code of length $mn$ over $\F$ such that 
\[\C = \{ (c_{ij})_{i \in [m], j \in [n]} | \forall i \  (c_{i1},\ldots,c_{in}) \in C_i; \forall j, (c_{1j},\ldots,c_{mj}) \in C'_j\}.\]
\end{definition}
In the above definition, when all the codes $C_i$ corresponding to the rows are equal and all the codes $C'_j$ corresponding to the columns are equal, we obtain a standard product code.

\section{Rate of Irregular Product Codes}\label{sec:Rate}

\begin{theorem}\label{rateTheorem}
Consider an $m \times n$ irregular product code $\C = \C(\{C_i\}_i,\{C'_j\}_j)$. 
Let $0 \le a_1 \le \ldots \le a_m \le n$ and $0 \le b_1 \le \ldots \le b_n \le m$ be two integer sequences.
For $i\in [m]$, assume that the value of the first $a_i$ coordinates of any codeword in $C_i$ can generate the remaining coordinates
(in the sense that the values of these remaining coordinates are a function of the values of the first $a_i$ coordinates). 
Similarly, for each $j \in [n]$, assume that the first $b_j$ coordinates of any codeword in $C'_j$ can generate the remaining coordinates. 
Then
\begin{enumerate}
 \item $\C$ has dimension at most 
\begin{equation}\label{eqn:dim}
  k_\C := \sum_{j=1}^n \sum_{i=b_{j-1}+1}^{b_j} \max(a_i - j + 1, 0),
\end{equation}
where we define $b_0 := 0$.
 \item If furthermore for all $i \in [m], j \in [n]$, $C_i$ is a linear code of dimension $a_i$ and
 $C'_j$ is a linear code of dimension $b_j$,
 and
 $C_1 \subseteq \cdots \subseteq C_m$ and $C'_1 \subseteq \cdots \subseteq C'_n$, 
then $\C$ has dimension exactly $k_\C$ as given by (\ref{eqn:dim}). 
\end{enumerate}
\end{theorem}
\begin{proof}
The coordinates of a codeword in $\C$ are all pairs $(i,j) \in [m] \times [n]$.
In the following, we will describe a procedure that returns some subset of these coordinates as ``generating coordinates''. As we shall see and as their name suggests, one can generate the remaining coordinates of a codeword in $\C$ from these coordinates. The number of these generating coordinates will be an upper bound on the dimension of $\C$. This is only an upper bound because there might be some settings of these generating coordinates that do not give rise to valid codewords.

In this procedure, initially all coordinates are unmarked. Each coordinate will eventually be marked either as ``generating'' or as ``determined''. A row (column) where not all coordinates have been marked is called ``available''. A row (column) whose marked coordinates can generate the values of the remaining unmarked coordinates is called ``determined''.

While there exists an unmarked coordinate
\begin{description}
\item[(A)] if there exists an available determined row
\begin{itemize}
\item pick the available determined row with the smallest index
\item mark its unmarked coordinates as ``determined''
\end{itemize}
 \item[(B)]   else if there exists an available determined column
\begin{itemize}
\item pick the available determined column with the smallest index
\item mark its unmarked coordinates as ``determined''
\end{itemize}
\item[(C)]        else
\begin{itemize}
\item pick the available row with the smallest index
  \item          starting from the smallest unmarked index,
            mark as many coordinates as is necessary as ``generating'' until
            the first $a_i$ coordinates are marked 
   \item         mark the remaining coordinates as ``determined''.
\end{itemize}
\end{description}

\begin{claim}\label{Claim1}
In the above procedure, the number of coordinates finally marked as generating are 
$$\displaystyle \sum_{j=1}^n \sum_{i=b_{j-1}+1}^{b_j} \max(a_i - j + 1, 0).$$
\end{claim}
\begin{proof}[Proof (of Claim~\ref{Claim1})]
For each row $i$, we count the number of generating coordinates in row $i$:

Assume $i > b_n$.
During the procedure, rows with smaller index become determined earlier and hence become fully marked earlier also.
Thus, before the procedure executes on row $i$, 
all rows $1, \ldots, i-1$ have been fully marked, 
hence all the columns are determined. 
Hence, 
the columns one-by-one cause the procedure to go through (B),
until row $i$ becomes determined,
at which point the procedure goes through (A) on row $i$.
Hence, no coordinate in row $i$ is ever going to be marked as generating.

Now assume $i \le b_n$ and consider the greatest $j$ such that $b_{j-1} < i$.
Before 
the procedure executes on row $i$,
all rows $1, \ldots, i-1$ have become fully marked,
hence all columns $1, \ldots, j-1$ are determined.
If $a_i < j$, all these columns cause the procedure to go through (B) one-by-one 
until row $i$ becomes determined,
at which point the procedure goes through (A) on row $i$. 
Hence, in this case, no coordinate in row $i$ is ever going to be marked as generating.
If, on the other hand, $a_i \ge j$ then
consider coordinate $(i,j)$.
Since the coordinates marked in a column are always a prefix of the column and since column $j$ requires $b_j \ge i$ marked coordinates to become determined,
the coordinate $(i, j)$ cannot be marked through (B) on column $j$ (instead of through row $i$).
By a similar argument, all coordinates $(i,j+1), \ldots, (i, n)$ are going to be marked through row $i$ (rather than through their columns).
This implies that when the procedure executes on row $i$,
coordinates $(i,j), \ldots, (i, a_i)$ are not yet marked,
and so the procedure goes through (C) on $i$, 
and exactly these coordinates are marked as generating.

In other words, the number of generating symbols in row $i$ is $\max(a_i - j + 1, 0)$.
This completes the proof of Claim~\ref{Claim1}.
\end{proof}

\vspace{3mm}

From the way the $k_\C$ generating coordinates were chosen, it is clear 
that the value of a codeword of $\C$ is a function of its value at these $k_\C$ coordinates. 
This finishes the proof of part 1 of Theorem~\ref{rateTheorem}.
To prove part 2 of Theorem~\ref{rateTheorem}, we show that under the conditions of part 2, 
the above procedure naturally gives rise to a systematic encoding algorithm for code $\C$:
When it marks a coordinate as generating, it can place an information symbol in this coordinate; when it marks a coordinate as determined while executing on a row (column), the value at this coordinate is generated from the generating coordinates of this row (column) according to the corresponding row (column) code.
We only need to show that any setting of the $k_\C$ generating coordinates gives rise to a valid codeword of $\C$. 

This algorithm begins with an empty $m \times n$ matrix $(c_{ij})$ corresponding to a codeword and fills its entries until all entries are filled and we have a matrix $(c_{ij}) \in \F^{m \times n}$.
To show that the final matrix $(c_{ij})$ is a valid codeword, 
we prove by induction on the number of steps of the algorithm that
$(c_{ij})$ never 
 \emph{violates} any row code or column code. 
By that we mean that for every row $i$ (column $j)$, 
at any point during the algorithm
the filled entries in row $i$ (column $j$) are 
a projection of a valid codeword in $C_i$ ($C'_j$) on these entries;
in other words, these filled entries
do not satisfy any linear constraint that is not satisfied by $C_i$ ($C'_j$).

No entry $c_{ij}$ will ever violate its row code. A proof of this claim goes as follows:
Since determined rows are given precedence over determined columns,
$c_{ij}$ is never filled through column $j$ if $j > a_i$. 
Indeed, if $j > a_i$, row $i$ must have been already determined at the point where $c_{ij}$ is filled.
It means that when row $i$ is picked by the algorithm, at most its first $a_i$ entries are filled.
Since $C_i$ is generated by its first $a_i$ coordinates and has dimension $a_i$, 
any setting of these coordinates will correspond to the projection of some valid codeword on its first $a_i$ coordinates. 

Thus, the only case we need to consider is that of an entry $c_{ij}$ violating its column code $C'_j$ (when row $i$ is being filled). We claim that this can also never happen. Let $c_{ij}$ be the first entry that violates its column code so that for all $i'< i$ and all $j' < j$, the entries $c_{i'j'}, c_{i'j}$, and $c_{ij'}$ do not violate their respective column codes. 
As $c_{ij}$ violates $C'_j$, we must have that $b_j < i$. Since $C'_j$ has
dimension $b_j$ and is generated by its first $b_j$ coordinates, there exists
$(\beta_1, \ldots, \beta_{i-1}) \in \F^{i-1}$ such that for any valid codeword
$(y_1, \ldots, y_n)$ of $C'_j$, we have
$y_i = \langle \beta, y_{1 \cdots i-1}\rangle$ but 
\begin{equation} \label{cij}
c_{ij} \ne \langle \beta, c_{1 \cdots i-1, j} \rangle.
\end{equation}

Since $C'_1 \subseteq \cdots \subseteq C'_j$, this implies that for each of the first $j-1$ columns, its first $i$ coordinates correspond to the projection of a valid codeword of $C'_j$ on its first $i$ coordinates.
Thus, for each $j' < j$, we have that 
\begin{equation}\label{cij'}
 c_{ij'} = \langle \beta, c_{1\cdots i-1,j'}\rangle.
\end{equation}

On the other hand, since $C_1 \subseteq \cdots \subseteq C_i$ and $j < a_i$, a similar argument shows that there exists a vector $\alpha \in \F^{j-1}$ such that for all $i' \leq i$, 
\begin{equation}\label{ci'j}
 c_{i'j} = \langle \alpha, c_{i',1\cdots j-1}\rangle.
\end{equation}
Using (\ref{cij'}) and (\ref{ci'j}), we see that $c_{ij} = \sum_{1 \le i' < i,1 \le j' < j} \alpha_{i'} \beta_{j'} c_{i'j'}$. But using (\ref{cij}) and (\ref{cij'}), we see that $c_{ij} \ne \sum_{1 \le i' < i, 1 \le j' < j} \alpha_{i'} \beta_{j'} c_{i'j'}$. 
This contradiction shows that 
no $c_{ij}$ violates a column code.
\end{proof}

\section{Minimum Distance of Irregular Product Codes}\label{sec:Dist}
The following theorem gives the best general lower bound on the minimum distance of an irregular product code
in terms of the minimum distances of the individual row and column codes.
Notice that this does not preclude the possibility of obtaining better lower bounds 
if we know more about the row and column codes.
\begin{theorem}\label{thm:MinDist}
For two integer sequences $n \ge d_1 \ge \ldots \ge d_m \ge 1$ and 
$m \ge d'_1 \ge \ldots \ge d'_n \ge 1$,
define $$ D = \min_{1 \le i \le m - d_j + 1; 1 \le j \le n - d_i + 1} \ 
\max_{i - 1 \le i' \le m;  j - 1 \le j' \le n} - (i' - i + 1) (j' - j + 1) +
 \sum_{k=i}^{i'} d_k + \sum_{k=j}^{j'} d'_k.$$
The number $D$ is the minimum weight of a binary nonzero $m \times n$ matrix 
where every nonzero row $i$ has weight $\ge d_i$
and every nonzero column $j$ has weight $\ge d'_j$.
Therefore, if $\C = \C(\{C_i\}_i,\{C'_j\}_j)$ is an $m \times n$ product code
such that mindist$(C_i) = d_i$ and mindist$(C'_j) = d'_j$,
then mindist$(\C) \ge D$.
On the other hand, for any two sequences $n \ge d_1 \ge \ldots \ge d_m \ge 1$ and 
$m \ge d'_1 \ge \ldots \ge d'_n \ge 1$, there exist row codes $C_i$ and column codes $C'_j$ 
with mindist$(C_i) = d_i$ and mindist$(C'_j) = d'_j$, such that
mindist$(\C) = D$.
\end{theorem}
\begin{proof}
Consider a minimum weight binary nonzero $m \times n$ matrix $M$
where every nonzero row $i$ has weight $\ge d_i$
and every nonzero column $j$ has weight $\ge d'_j$.
Because the sequence of $d_i$s and the sequence of $d'_j$s are sorted nonincreasingly,
we can permute the rows and columns of $M$
in such a way that 
the nonzero rows become rows $i, i + 1, \ldots, m$ for some $1 \le i \le m$
and the nonzero columns become columns $j, j + 1, \ldots, n$ for some $1 \le j \le n$
while still preserving the property that 
for each (nonzero) row $i'' \in [i,m]$ (column $j'' \in [j,n]$)  has weight $\ge d_{i''}$ ($\ge d_{j''}$).
Now, consider the submatrix of $M$ consisting of the intersection of rows $i, \ldots, m$ and columns $j, \ldots, n$. 
We want to minimize the number of ones in this submatrix.
Instead we look at
the problem of maximizing the number of zeros in this submatrix, which can be expressed 
as a max-flow problem where for each $i'' \in [i, m]$ there is an edge of capacity $n - j + 1 - d_{i''}$ from the source to a vertex that corresponds to row $i''$, 
for each $j'' \in [j,n]$ there is an edge of capacity $m - i + 1 - d'_{j''}$ from a vertex that corresponds to row $j''$ to the sink, and there is an edge of capacity 1 from each row $i''$ to each row $j''$.
Using the fact that the min-cut equals max-flow,
one can show that the minimum number of ones in this submatrix 
is indeed 
$$\max_{i - 1 \le i' \le m, \ j - 1 \le j' \le n} \ \  \sum_{k=i}^{i'} d_k + \sum_{k=j}^{j'} d'_k - (i' - i + 1) (j' - j + 1).$$
Since row $i$ has at least $d_i$ ones, we have $d_i \le n - j+ 1$. Similarly, $d_j \le  m - i + 1$.
Minimizing over all $i$ and $j$, we get that $D$ is the weight of $M$.

Now, we can deduce that for any two distinct codewords in the product code $\C$, since they differ on at least $d_i$ ($d'_j$) coordinates in every row $i$ (column $j$) in which they differ in at least one coordinate, they have Hamming distance $\ge D$.

Finally, assume the sequences $d_1, \ldots, d_m$ and $d'_1, \ldots, d'_n$ are given.
Find a weight-$D$ matrix $M \in \{0,1\}^{m \times n}$ where each nonzero row $i$ has weight $\ge d_i$ and 
each nonzero column $j$ has weight $\ge d'_j$.
For each zero row $i$, we define $C_i$ to be any linear code of minimum distance $d_i$.
Similarly, for each nonzero row $i$, we want to find a code $C_i$
of minimum distance $d_i$
such that row $i$ of matrix $M$ is a codeword in $C_i$.
An $[n,k=n-d_i+1, d_i]$-Reed-Solomon code has at least one codeword of weight $w$ for each $w \in [d_i, n]$ (because the degree-$(k-1)$ polynomial $(x - \alpha_1)^{k + w - n} (x - \alpha_2) \ldots (x - \alpha_{n-w})$ has exactly $w$ non-roots among distinct elements $\alpha_1, \ldots, \alpha_n$ of a field.)
Thus,  we can multiply each codeword coordinate of such a Reed-Solomon code by an appropriate nonzero field element in such a way that row $i$ of the zero-one matrix $M$ is a codeword in the resulting code $C_i$ of minimum distance $d_i$.
Similarly, 
For each zero column $j$, we define $C'_j$ to be any linear code of minimum distance $d'_j$.
Similarly, for nonzero columns $j$, we can find column codes $C'_j$ of appropriate minimum distance $d'_j$
such that row $j$ of matrix $M$ is a codeword in $C'_j$.
Finally, we need to choose the same symbol field for all these codes $C_i$ and $C'_j$. 
We can choose the field to be $\F_q$ for some $q \ge \max(m,n)$.
Then, the minimum distance of $\C = \C(\{C_i\}_i,\{C'_j\}_j)$ is $D$.

\end{proof}

\section{Asymptotic Analysis of Decoding Irregular Product Codes
on Erasure Channels}\label{asymptoticSection}
We need the following definition for the next theorem.
\begin{definition}\label{def:asymptotic}
Consider an $m \times n$ irregular product code $\C = \C(\{C_i\}_i,\{C'_j\}_j)$.
We are interested in the asymptotic behavior of $\C$, therefore
we think of $\C$ not individually but as one member of a family of irregular product codes
where $m$ and $n$ grow.
Suppose that 
$\alpha, \beta: [0,1] \rightarrow [0,1]$ are non-decreasing real functions.
We say that the row and column codes have asymptotic normalized minimum distance distribution $\alpha$ and $\beta$ if 
for every 
$\delta_1, \delta_2 > 0$, for large enough $m$ and $n$, for
each $i \in [m], j \in [n]$ we have $|\mbox{mindist}(C_i)/n - \alpha(x)| \le \delta_1$ for some $x$ such that $|1 - i/m - x| \le \delta_2$
and 
$|\mbox{mindist}(C'_j)/m - \beta(y)| \le \delta_1$ for some $y$ such that $|1 - j/n - y| \le \delta_2$.

\end{definition}
\begin{theorem}\label{thm:asymptoticDecoding}
Assume an $m \times n$ product code $\C = \C(\{C_i\}_i,\{C'_j\}_j)$ 
having asymptotic normalized minimum distance distribution $\alpha$ and $\beta$ as in Definition \ref{def:asymptotic}.
Assume that neither of $m$ or $n$ grows exponentially or faster in terms of the other one.
Consider that a codeword in $\C$ is sent over an erasure channel 
where each symbol is erased with probability $\epsilon > 0$.
We iteratively decode row codes and column codes of $\C$
whenever the number of erasures in a row or column is smaller than the minimum distance 
of the code corresponding to that row or column.
Assume that 
\begin{equation} \alpha^{-1}(\epsilon \beta^{-1} (\epsilon x)) < x \mbox{\ for all} \ x \in (0,1], \label{densityEvolutionInequality}\end{equation}
where we define $\beta^{-1}(x) = \sup(S_x)$ for $S_x = \{z \in [0,1]: \beta(z) \le x\}$ if $S_x \ne \emptyset$ and we define $\beta^{-1}(x) = 0$ if $S_x = \emptyset$. We define $\alpha^{-1}$ similarly.
Then for any constant $\delta_0 > 0$, for large enough codes in the family, 
all except a $\delta_0$-fraction of the symbols can be decoded 
except with a probability exponentially small in $\min(m,n)$. 

\end{theorem}
\begin{proof}
Let $y = \beta^{-1}(\epsilon)$. Notice that $y$ only depends on $\beta$ and does not depend on $m$ and $n$. Let $1 \ge y' > y$, where we assume that $y'$ also is a number independent of $m$ and $n$.
We claim that for large enough $m$ and $n$, with very high probability all except the last $y'$-fraction of the columns can be decoded in the first step. 
To see this, let $y'' \in (y, y')$. We know $\beta(y'') > \epsilon$, hence for some $\epsilon' > \epsilon$, for large enough codes in the family, we have mindist$(C'_j) \ge \epsilon' m$ when $n - j \le y' n$, i.e., for the last $y'$-fraction of the columns. 
The probability that each of the length-$m$ codes corresponding to these columns cannot be decoded is exponentially small in $m$ by the Chernoff bound, since these codes can decode up to $\epsilon' m - 1$ erasures while we have on average $\epsilon m$ erasures. Since the number of columns does not grow exponentially in $m$, we can use a union bound to derive our claim.

Next we define $x_1 = \alpha^{-1}(\epsilon y)$. We claim that for any $1 \ge x' > x_1$ with very high probability, all except the last $x'$ fraction of rows can be decoded. To see this, let $x'' \in (x_1, x')$. 
We know $\alpha(x'') > \epsilon y$, hence $\alpha(x'') > \epsilon y''$ for some $y'' > y$. 
We can conclude that mindist$(C_i) \ge \epsilon y'' n$ when $m - i \le x' m$, i.e., for the last $x'$-fraction of the columns.
On the other hand, if we choose $y' \in (y, y'')$, by the previous paragraph with high probability all  the symbols not appearing in the last $y'$-fraction of the columns are decoded for large enough codes. Therefore, the average number of undecoded symbols at each row is at most $\epsilon y' n$. Again, we can derive our claim by a union bound on Chernoff bounds.

Repeating the above argument back and forth between rows and columns, we get a non-increasing sequence $x_0 = 1, x_1, x_2, \ldots$ where $x_{i+1} = \alpha^{-1}(\epsilon \beta^{-1} (\epsilon x_i))$.
Here $1 - x_i$ denotes the approximate fraction of rows that are guaranteed to be decoded after $i$ back-and-forth rounds of decoding.
If this sequence converges to 0, then $\delta_0 > x_i$ for some $i$. 
That would mean that with high probability, at most a $\delta_0$-fraction of the rows and hence at most a $\delta_0$-fraction of all the symbols are not decoded by the end of the algorithm. 

So we just need to check that the monotonic sequence $x_0, x_1, x_2, \ldots$ converges to 0 if condition (\ref{densityEvolutionInequality}) is satisfied.
Assume otherwise that $x^* = \lim_{i \rightarrow \infty} x_i > 0$. 
We have $$\alpha^{-1}(\epsilon \beta^{-1} (\epsilon x^*)) = \alpha^{-1}(\epsilon \beta^{-1} (\epsilon \lim_{i \rightarrow \infty} x_i)) = \lim_{i \rightarrow \infty} \alpha^{-1}(\epsilon \beta^{-1} (\epsilon x_i)) = x^*$$
because $\alpha^{-1}$ and $\beta^{-1}$ can be shown to be right-continuous.
This contradicts condition (\ref{densityEvolutionInequality}).
\end{proof}

\section{Irregular Product Codes from MDS Codes}\label{sec:Cap}
\begin{proposition}\label{rateIrregularProductCodesfromMDS}
Consider an irregular product code 
$\C = \C(\{C_i\}_i,\{C'_j\}_j)$ 
where
$C_i$ is an $[n,a_i]$-MDS code and $C'_j$ is an $[m, b_j]$-MDS code for all $i, j$.
If $a_1, \ldots, a_m$ and $b_1, \ldots, b_n$ are non-decreasing sequences,
then the dimension of $C$ is upper-bounded by formula (\ref{eqn:dim}).

Furthermore, given any two integer sequences $0 \le a_1 \le \ldots \le a_m \le n$ and $0 \le b_1 \le \ldots \le b_n \le m$, we can meet this upper-bound in the following way.
Choose $n$ distinct elements $\alpha_1, \ldots, \alpha_n$ and $m$ distinct elements $\beta_1, \ldots, \beta_m$ of the symbol field $\F$.
Let $V$ be the $a_m \times n$ Vandermonde matrix $V_{ij} = \alpha_j^{i-1}$ and $V'$ be the $b_n \times m$ Vandermonde matrix $V'_{ij} = {\beta_j}^{i-1}$.
Let $C_i$ 
be the Reed-Solomon code having as generator matrix the first $a_i$ rows of the matrix $V$ and 
$C'_j$ be the Reed-Solomon code having as generator matrix the first $b_j$ rows of $V'$. 
Then the dimension of $C$ is given exactly by formula (\ref{eqn:dim}).
\end{proposition}
\begin{proof}
Any $a_i$ coordinates of $C_i$ (and in particular the first $a_i$ coordinates) can generate the rest of the coordinates. Similarly, the first $b_j$ coordinates of $C'_j$ generate the rest of the coordinates. Therefore, by Theorem~\ref{rateTheorem}, the rate of the irregular product code is upper-bounded by formula (\ref{eqn:dim}).

Now assume $C_i$ has as generator matrix the first $a_i$ rows of the Vandermonde matrix $V$.
Then for $i < i'$ we have $a_i \le a_{i'}$, hence $C_i$ has a generator matrix which is a submatrix of a generator matrix of $C_{i'}$, and hence $C_i$ is a subcode of $C_{i'}$.
We can argue similarly about the column codes. 
This shows that by Theorem~\ref{rateTheorem}, the exact rate is given by formula (\ref{eqn:dim}).
\end{proof}

\begin{theorem}\label{genericConstruction}
For each $\epsilon > 0$, the following is a generic way of constructing families of irregular product codes with asymptotic rate $1 - \epsilon$ such that for any constant $\delta > 0$ one can decode almost all the symbols of a codeword sent over an erasure channel having erasure probability $\epsilon - \delta$:

Choose any non-decreasing function $\beta:[0,1]\rightarrow[0,1]$ with $\beta(1) \le \epsilon$
and $\lim_{y \rightarrow 0} \beta(y) = 0$.
Define $\alpha:[0,1]\rightarrow[0,1]$ by $\alpha(x) = \epsilon \beta^{-1}(\epsilon x)$
where $\beta^{-1}$ is defined in Theorem~\ref{thm:asymptoticDecoding}.
Choose $m$ and $n$ as you wish but neither of $m$ or $n$ should grow exponentially or faster in the other one.
Then choose
$0 \le a_1 \le \ldots \le a_m \le n$ and $0 \le b_1 \le \ldots \le b_n \le m$
as you wish but in such a way that $a_i = n (1 - \alpha(1 - i/m + o(1)) + o(1))$
and $b_j = m (1 - \beta(1 - j/n + o(1)) + o(1))$.
Then choose the row codes $C_i$ to be nested linear MDS codes of dimension $a_i$ (for example as in Proposition \ref{rateIrregularProductCodesfromMDS}). Choose similarly column codes $C'_j$ of dimension $b_j$.
\begin{figure}
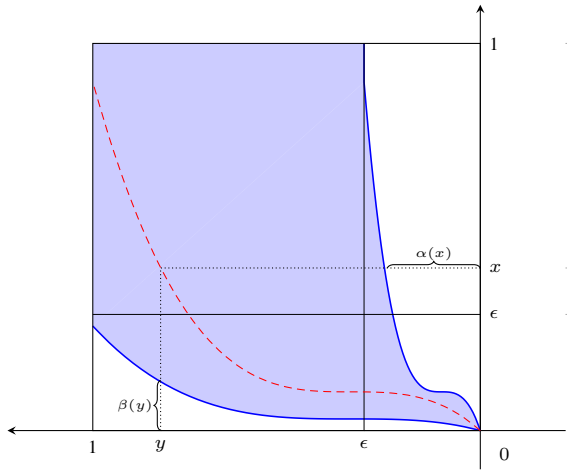

\centering
\inputtikz{asym}
\caption{\label{fig:as}This figure shows how the curve for $\alpha$ is obtained from the curve for $\beta$ in Theorem~\ref{genericConstruction}. We stretch the curve for $\beta$ vertically by a factor of $1/\epsilon$, and then shrink the curve for $\alpha$ horizontally by a factor of $\epsilon$. That is, whenever $x = \beta(y)/\epsilon$, we have $\alpha(x) = \epsilon y$. 
The area of the shaded region denotes the asymptotic rate of the code, which is $1 - \epsilon$.}
\end{figure}
\end{theorem}
\begin{proof}
For MDS codes $C_i$, we have mindist$(C_i) = n - a_i + 1$.
Therefore the minimum distance of the row codes can be approximated correctly
using $\alpha$ as
in the statement of Theorem~\ref{thm:asymptoticDecoding}.
We can do similarly for the column codes.
In order to show we can decode from $\epsilon - \delta > 0$ fraction of errors, 
it is enough by Theorem~\ref{thm:asymptoticDecoding}
to check that $$\alpha^{-1}((\epsilon - \delta)\beta^{-1}((\epsilon - \delta)x)) \ge x$$ 
does not happen for any $x \in (0, 1]$. If it does for some $x$, then for any $0 \le x' < x$, we have $$\epsilon \beta^{-1}(\epsilon x') = \alpha(x') \le (\epsilon - \delta)\beta^{-1}((\epsilon - \delta)x).$$
Since $\beta^{-1}(\epsilon x') \ge \beta^{-1}((\epsilon - \delta) x)$
for all $x' \in [(\epsilon - \delta)x/\epsilon, x)$,
this implies $\beta^{-1}(\epsilon x') = 0$ for all such $x'$.
This implies $\beta^{-1}(\epsilon x') = 0$ for all $x' \in (0, x)$ and this contradicts 
$\lim_{y \rightarrow 0} \beta(y) = 0$.

Now we want to calculate the asymptotic rate. By Theorem \ref{rateIrregularProductCodesfromMDS}, the dimension of the code is given by formula (\ref{eqn:dim}), which can be expressed 
as $$ mn \ \mathbb{E}_i[ \max(\frac{a_i - \max(\{j: b_j < i\})}{n}, 0)].$$ 
Thus, the rate is asymptotically equal to 
\begin{eqnarray*}
\int_{x =0}^1 \max(\beta^{-1}(x) - \alpha(x), 0) \ dx & = & \int_{x = 0}^1 (\beta^{-1}(x) - \epsilon \beta^{-1}(\epsilon x)) \ dx \\
& = & \int_{x=0}^1 \beta^{-1}(x) \ dx - \epsilon \int_{x=0}^1 \beta^{-1}(\epsilon x) \ dx \\
& = & 
[(1 - \epsilon) + \int_{x=0}^\epsilon \beta^{-1}(x) \ dx] - \int_{x=0}^\epsilon \beta^{-1}(x) \ dx \\
& = &1 - \epsilon.
\end{eqnarray*}
\end{proof}

We note that the only regular products codes based on MDS codes which have
decoding properties asymptotically as good as those constructed in
Theorem~\ref{genericConstruction} are codes where either $\alpha = 0$ or $\beta
= 0$.
These are regular product codes in which the row codes or column codes have rate
$1 - o(1)$.

\section{Examples of Finite-Length Irregular Product Codes}\label{sec:Finite}
In order to find an example of an irregular product code for finite but not so
small lengths, say $50 \times 50$,
we used the asymptotic irregular product code shown in
Figure~\ref{fig:straight-Line-asymptotic}
obtained from Theorem~\ref{genericConstruction}, in which $\alpha(x) =
\epsilon x$ and $\beta(y) = \epsilon y$ where $\epsilon$ is the erasure
probability. 
The area of the shaded region, which represents the systematic part of the code,
is the rate $1 - \epsilon$ of the code.


Next we slightly tuned the asymptotic code to a $50 \times 50$
irregular product code in such a way that 
\begin{itemize}
 \item the code
can start decoding better, by increasing the number of row and
column codes
having the highest minimum distance by a few;
\item more importantly, the code has a much higher probability of decoding all
symbols once most of the symbols have been decoded,
by forcing that the minimum distances of all row and column codes are 
at least some positive number, in this case 3.
\end{itemize}
We chose all row codes and all column codes to be nested MDS codes according to
Theorem~\ref{rateIrregularProductCodesfromMDS}.
The resulting code has a systematic part which is shown in
Figure~\ref{fig:50x50systematic}.

\begin{figure}
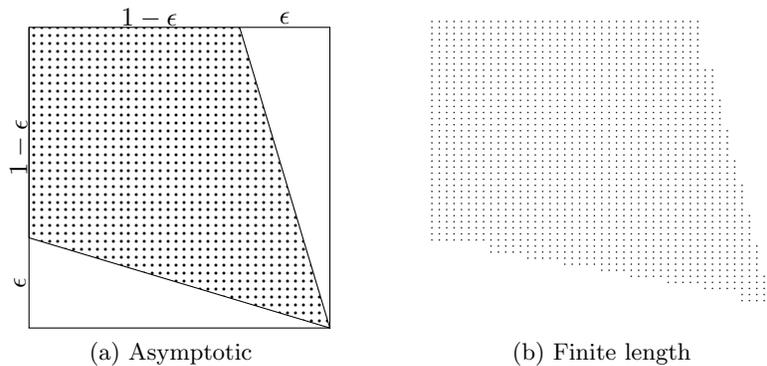

  \centering
\subfloat[Asymptotic]{\label{fig:straight-Line-asymptotic}\inputtikz{asymline}{}
}
\qquad
  \subfloat[Finite length]{\label{fig:50x50systematic}\inputtikz{50x50}{}}
  \caption{a) The shaded region corresponds to
the systematic or information part of the
code by choosing $\alpha(x) =
\epsilon x$ and $\beta(y) = \epsilon y$ where $\epsilon$ is the erasure
probability. b) Systematic part of a $[2500,1709]$ irregular
product}
  \label{fig:line}
\end{figure}


This code is a $[2500,1709]$ code of rate $0.6836$. We
compared this code to all regular product codes having rates $[0.6708, 0.684]$.
Note that most of these codes have rate even lower than this code. 
The result of the simulation is shown in Figure~\ref{fig:sim50x50}.
This plot shows the block/word error rate of the code in an erasure channel with
erase probability $\epsilon$. All constituent row and column codes in
irregular and regular cases are considered to be MDS with the corresponding
dimensions. Each point of these curves is obtained by $10^6$ simulations. The
erasure patterns for different values of $\epsilon$ have been coupled such that
the block error versus erasure probability curve is monotonic. One can see that
this code outperforms all product codes having lower rates.

\begin{figure}
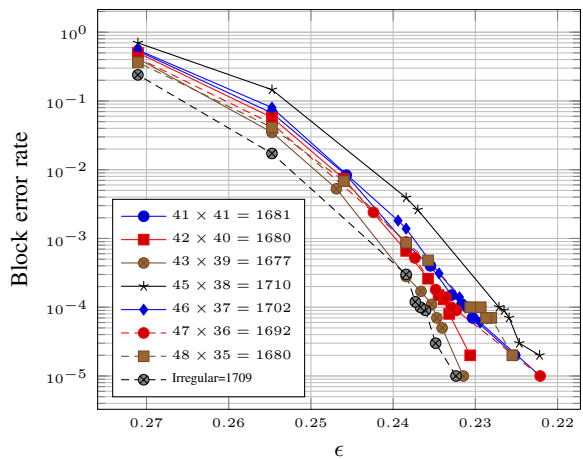

\centering
\inputtikz{n50r68}
\caption{Comparing an irregular $[2500,1709]$ code to
almost equal rate regular ones. The number on the legends indicate
the corresponding row and column dimensions.\label{fig:sim50x50}}
\end{figure}

Figure~\ref{fig:8x8} shows another case where an irregular code outperforms a
regular code for a much smaller length. We compared a regular~$[8 \times 8, 4 \times 7]$ product code
 with an irregular code which is shown in
Figure~\ref{fig:8x8K28shape}. Numbers on rows and columns indicate the
dimension of the corresponding row and column MDS code. The block error
probability of these codes is shown in Figure~\ref{fig:8x8K28ber}. Both codes
are $[64,28]$ codes.
\begin{figure}
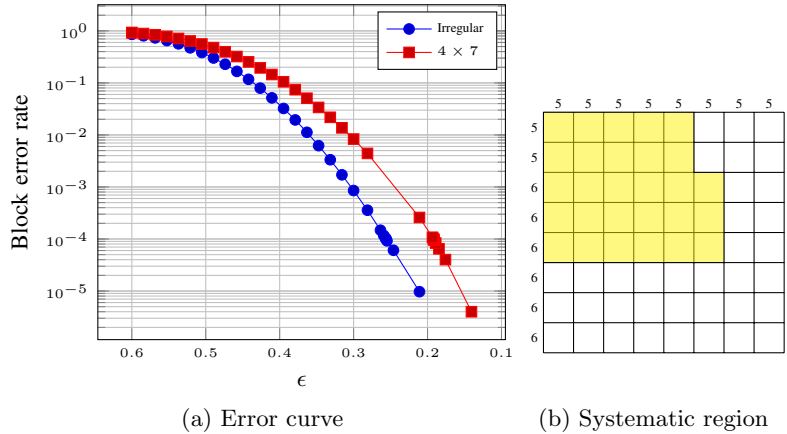

  \centering
  \subfloat[Error curve]{\label{fig:8x8K28ber}\inputtikz{8x8K28}{}}
  \subfloat[Systematic region]{\label{fig:8x8K28shape}\inputtikz{8x8K28shape}{}}
  \caption{Comparing an $8 \times 8$ regular and irregular product codes both
with dimension $28$}
  \label{fig:8x8}
\end{figure}


\bibliographystyle{plain}
\bibliography{ipc}

\begin{thebibliography}{10}

\bibitem{baldi_class_2009}
M.~Baldi, G.~Cancellieri, and F.~Chiaraluce.
\newblock A class of {Low-Density} {Parity-Check} product codes.
\newblock In {\em Advances in Satellite and Space Communications, 2009.
  {SPACOMM} 2009. First International Conference on}, pages 107 --112, July
  2009.

\bibitem{blankenship_block_2005}
Y.~Blankenship, B.~Classon, and V.~Desai.
\newblock Block product code design with the aid of union bounds.
\newblock In {\em Vehicular Technology Conference, 2005. {VTC} {2005-Spring.}
  2005 {IEEE} 61st}, volume~3, pages 1533 -- 1537 Vol. 3, June 2005.

\bibitem{cao_novel_2003}
Lei Cao and Chang~Wen Chen.
\newblock A novel product coding and recurrent alternate decoding scheme for
  image transmission over noisy channels.
\newblock {\em Communications, {IEEE} Transactions on}, 51(9):1426 -- 1431,
  September 2003.

\bibitem{chaichanavong_tensor-product_2006}
P.~Chaichanavong and {P.H.} Siegel.
\newblock Tensor-product parity code for magnetic recording.
\newblock {\em {IEEE} Transactions on Magnetics}, 42(2):350--352, February
  2006.

\bibitem{chiaraluce_extended_2004}
F.~Chiaraluce and R.~Garello.
\newblock Extended hamming product codes analytical performance evaluation for
  low error rate applications.
\newblock {\em {IEEE} Transactions on Wireless Communications},
  3(6):2353--2361, November 2004.

\bibitem{elias_error-free_1954}
P.~Elias.
\newblock Error-free coding.
\newblock {\em {IEEE} Transactions on Information Theory}, 4(4):29--37,
  September 1954.

\bibitem{kschischang_product_2003}
Frank~R. Kschischang.
\newblock Product codes.
\newblock In {\em Wiley Encyclopedia of Telecommunications}. John Wiley \&
  Sons, Inc., Hoboken, {NJ}, {USA}, April 2003.

\bibitem{lentmaier_product_2010}
M.~Lentmaier, G.~Liva, E.~Paolini, and G.~Fettweis.
\newblock From product codes to structured generalized {LDPC} codes.
\newblock In {\em Communications and Networking in China {(CHINACOM)}, 2010 5th
  International {ICST} Conference on}, pages 1 --8, August 2010.

\bibitem{pyndiah_near-optimum_1998}
{R.M.} Pyndiah.
\newblock Near-optimum decoding of product codes: block turbo codes.
\newblock {\em {IEEE} Transactions on Communications}, 46(8):1003--1010, August
  1998.

\bibitem{rankin_single_2001}
{D.M.} Rankin and {T.A.} Gulliver.
\newblock Single parity check product codes.
\newblock {\em Communications, {IEEE} Transactions on}, 49(8):1354 --1362,
  August 2001.

\bibitem{rankin_asymptotic_2003}
{D.M.} Rankin, {T.A.} Gulliver, and {D.P.} Taylor.
\newblock Asymptotic performance of single parity-check product codes.
\newblock {\em Information Theory, {IEEE} Transactions on}, 49(9):2230 -- 2235,
  September 2003.

\bibitem{rosnes_stopping_2008}
Eirik Rosnes.
\newblock Stopping set analysis of iterative {Row-Column} decoding of product
  codes.
\newblock {\em {IEEE} Transactions on Information Theory}, 54(4):1551--1560,
  April 2008.

\bibitem{stankovic_joint_2002}
V.~Stankovic, R.~Hamzaoui, and Zixiang Xiong.
\newblock Joint product code optimization for scalable multimedia transmission
  over wireless channels.
\newblock In {\em Multimedia and Expo, 2002. {ICME} '02. Proceedings. 2002
  {IEEE} International Conference on}, volume~1, pages 865 -- 868 vol.1, 2002.

\bibitem{tanner_recursive_1981}
R.~Tanner.
\newblock A recursive approach to low complexity codes.
\newblock {\em {IEEE} Transactions on Information Theory}, 27(5):533--547,
  September 1981.

\end{thebibliography}
\end{document}